\documentclass[a4paper,UKenglish]{lipics-v2018}
\usepackage{microtype}
\usepackage{amssymb,amsfonts,amsmath}
\usepackage{cite}
\usepackage{graphicx}
\usepackage{array}
\usepackage{algorithm}
\usepackage{algorithmicx}
\usepackage[noend]{algpseudocode}

\def\u{\mathbf{u}}
\def\v{\mathbf{v}}
\def\w{\mathbf{w}}
\def\T{\mathbf{T}}

\newtheorem{observation}{Observation}

\hideLIPIcs
\nolinenumbers

\title{Transition Property For Cube-Free Words}

\titlerunning{Transition Property For Cube-Free Words}

\author{Elena A. Petrova}{Ural Federal University, Ekaterinburg, Russia}{elena.petrova@urfu.ru}{}{Supported by the Russian Science Foundation, grant 18-71-00043}

\author{Arseny M. Shur}{Ural Federal University, Ekaterinburg, Russia}{arseny.shur@urfu.ru}{}{Supported by the Russian Ministry of Education and Science, project 1.3253.2017}

\authorrunning{E. A. Petrova and A. M. Shur}

\Copyright{Elena A. Petrova and Arseny M. Shur}

\subjclass{ Mathematics of computing $\rightarrow$ Discrete mathematics $\rightarrow$ Combinatorics $\rightarrow$ Combinatorics on words}


\begin{document}

\maketitle

\begin{abstract}
We study cube-free words over arbitrary non-unary finite alphabets and prove the following structural property: for every pair $(u,v)$ of $d$-ary cube-free words, if $u$ can be infinitely extended to the right and $v$ can be infinitely extended to the left respecting the cube-freeness property, then there exists a ``transition'' word $w$ over the same alphabet such that $uwv$ is cube free. The crucial case is the case of the binary alphabet, analyzed in the central part of the paper.

The obtained ``transition property'', together with the developed technique, allowed us to solve cube-free versions of three old open problems by Restivo and Salemi. Besides, it has some further implications for combinatorics on words; e.g., it implies the existence of infinite cube-free words of very big subword (factor) complexity. 
\end{abstract}

\section{Introduction}

The concept of power-freeness is in the center of combinatorics on words. This concept expresses the restriction on repeated blocks (factors) inside a word: an $\alpha$-power-free word contains no block which consecutively occurs in it $\alpha$ or more times. For example, the block \texttt{an} in the word \texttt{banana} is considered as having $5/2$ consecutive occurrences; thus the word \texttt{banana} is 3-power-free (cube-free) but not $(5/2)$-power-free or $2$-power-free (square-free); the block \texttt{mag} in the word \texttt{magma} occurs consecutively $5/3$ times, and so \texttt{magma} is square-free. Power-free words and languages are studied in lots of papers starting with the seminal works by Thue \cite{Thue06, Thue12}, who proved, in particular, the infiniteness of the sets of binary cube-free words and ternary square-free words. However, many phenomena related to power-freeness are still not understood. 

One group of problems about power-free words concerns their structure and extendability. In 1985, Restivo and Salemi presented \cite{ReSa85a} a list of five problems, originally considered only for ternary square-free words and binary overlap-free words, but equally important for every power-free language. Suppose that a finite alphabet $\Sigma$ is fixed and we study $\alpha$-power-free words over $\Sigma$. Here are the problems.\\
\emph{Problem 1.} Given an $\alpha$-power-free word $u$, decide whether there are infinitely many $\alpha$-power-free words having (a) the prefix $u$; or (b) the suffix $u$; or (c) the form $vuw$, where $v$ and $w$ have equal length. (Such words $u$ are called, respectively, right extendable, left extendable, and two-sided extendable.)\\ 
\emph{Problem 2.} Given an $\alpha$-power-free word $u$, construct explicitly an  $\alpha$-power-free infinite word having $u$ as prefix, provided that $u$ is right extendable.\\
\emph{Problem 3.} Given an integer $k\ge 0$, does there exist an  $\alpha$-power-free word $u$ with the properties (i) there exists a word $v$ of length $k$ such that $uv$ is $\alpha$-power free and (ii) for every word $v'$ of bigger length, $uv'$ is not $\alpha$-power free.\\
\emph{Problem 4.} Given two $\alpha$-power-free words $u$ and $v$, decide whether there is a ``transition'' from $u$ to $v$ (i.e., does there exist a word $w$ such that $uwv$ is $\alpha$-power free).\\
\emph{Problem 5.} Given two $\alpha$-power-free words $u$ and $v$, find explicitly a transition word $w$, if it exists.

These natural problems appear to be rather hard. Only for Problem~1a,b there is a sort of a general solution: a backtracking decision procedure exists for all $k$-power-free languages, where $k\ge2$ is an integer \cite{Cur95,CuSh03}. In a number of cases, the parameters of backtracking were found by computer search, so it is not clear whether this technique can be extended for $\alpha$-power-free words with rational $\alpha$. The decision procedure also gives no clue to Problem~2. 

There is a particular case of binary overlap-free words, for which all problems are solved in \cite{Car84,ReSa85a} (more efficient solutions were given in \cite{Car93}). These words have a regular structure deeply related to the famous Thue-Morse word, and it seems that all natural algorithmic problems for them are solved. For example, the asymptotic order of growth for the binary overlap-free language is computed exactly \cite{GuPr13,JPB09}, and even the word problem in the corresponding syntactic monoid has a linear-time solution \cite{Sh12sem}. Most of the results can be extended, with additional technicalities, to binary $\alpha$-power-free words for any $\alpha\le 7/3$, because the structure of these words is essentially the same as of overlap-free words (see, e.g., \cite{KaSh04}). 
However, the situation changes completely if we go beyond the polynomial-size language of binary $(7/3)$-power-free words. In the exponential-size $\alpha$-power-free languages\footnote{For $\alpha> 7/3$, the language of binary $\alpha$-power-free words has exponential size \cite{KaSh04}. \textit{The exponential conjecture} says that for $k\ge 3$ \textit{all} infinite power-free languages over $k$ letters have exponential size. This conjecture is proved for $k\le 10$ \cite{KoRa11,Och06} and odd $k$ up to 101 \cite{TuSh12}.} the diversity of words is much bigger, so it becomes harder to find a universal decision procedure. The only results on Problems~1-5 apart from those mentioned above are the positive answers to Problem~3 (including its two-sided analog) for the two classical test cases: for ternary square-free words \cite{PeSh12} and for binary cube-free words \cite{PeSh12a}. 

In this paper, we study cube-free words over arbitrary alphabets. Still, the crucial case is the  one of the binary alphabet; the central part of the paper is the proof of the following transition property of binary cube-free words.

\begin{theorem} \label{t:uwv}
For every pair $(u,v)$ of binary cube-free words such that $u$ is right extendable and $v$ is left extendable, there exists a binary word $w$ such that $uwv$ is cube free.
\end{theorem}

After proving Theorem~\ref{t:uwv} in Section~3, we use it and its proof to derive further results. In Section~4 we prove the transition property for  arbitrary alphabets (Theorem~\ref{t:uwv3}), while in Section~5 we use this property to solve Restivo-Salemi Problems~2,~4, and 5. Thus, all Restivo--Salemi problems for binary cube-free words are solved; this is the first fully solved case since the original publication of the problems. For cube-free words over bigger alphabets, only Problem~3 is not yet solved.

We finish the introduction with two remarks. First, the result of Theorems~\ref{t:uwv} and~\ref{t:uwv3} was conjectured, in a slightly weaker form, for all infinite power-free languages \cite[Conj 1]{Sh09dlt}. This conjecture is related to the properties of finite automata recognizing some approximations of power-free languages and was supported by extensive numerical studies. The transition words can be naturally interpreted as transitions in those automata and the transition property forces the automata to be strongly connected. Second, recently it was shown \cite[Thm 39]{ShSh18} that the transition property implies the existence of infinite $\alpha$-power-free words of very big subword complexity. Namely, Theorems~\ref{t:uwv} and~\ref{t:uwv3} imply that for every $d\ge 2$ there exists a $d$-ary cube-free infinite word which contains \textit{all} two-sided extendable $d$-ary cube-free finite words as factors.

\section {Preliminaries}

\subsection{Notation and Definitions}
By default, we study words over finite alphabets $\Sigma_d$ of cardinality $d\ge2$, writing $\Sigma_d=\{a,b,c_1,\ldots,c_{d-2}\}$ (mostly we work with $\Sigma_2=\{a,b\}$).  Standard notions of factor, prefix, and suffix are used. The set of all finite (nonempty finite, infinite) words over an alphabet $\Sigma$ is denoted by $\Sigma^*$ (resp., $\Sigma^+$, $\Sigma^\infty$). Elements of $\Sigma^+$ ($\Sigma^\infty$) are treated as functions $w:\{1,\ldots,n\}\to \Sigma$ (resp., $\w:{\mathbb N}\to\Sigma$). We write $[i..j]$ for the range $i,i{+}1,\ldots,j$ of positive integers; the notation $w[i..j]$ stands for the factor of the word $w$ occupying this range as well as for the particular occurrence of this factor in $w$ \textit{at position} $i$. Note that $w[i..i]=w[i]$ is just the $i$th letter of $w$. Let  $w[i_1..j_1]$ and $w[i_2..j_2]$ be two factors of $w$. If the ranges $[i_1..j_1]$ and $[i_2..j_2]$ have a nonempty intersection, their intersection and union are also ranges; we refer to the factors of $w$, occupying these ranges, as the \textit{intersection} and the \textit{union} of $w[i_1..j_1]$ and $w[i_2..j_2]$. The word $\overleftarrow{w}=w[n]\ldots w[1]$ is called the \emph{reversal} of the word $w$ of length $n$. 

We write $\lambda$ for the empty word and $|w|$ for the length of a word $w$ (infinite words have length $\infty$). A word $w$ has \emph{period} $p<|w|$ if $w[1..|w|{-}p]=w[p{+}1..|w|]$; the prefix $w[1..p]$ of $w$ is the \textit{root} of this period of $w$. One of the most useful properties of periodic words is the following.

\begin{lemma}[Fine, Wilf \cite{FiWi65}] \label{fw}
If a word $u$ has periods $p$ and $q$ and $|u|\ge p+q-\gcd(p,q)$ then $u$ has period $\gcd(p,q)$.
\end{lemma}

A \emph{cube} is a nonempty word of the form $uuu$, also written as $u^3$; we refer to $u$ as the root of this cube. A word is \emph{cube-free} (\emph{overlap-free}) if it has no cubes as factors (resp., no factors of the form $cwcwc$, where $c$ is a letter). There exist binary overlap-free (and thus cube-free) infinite words \cite{Thue12}.

Let $\Sigma_d$ be fixed. A word $w\in\Sigma_d^*$ is called a \textit{right context} of a cube-free word $u\in\Sigma_d^*$ if $uw$ is cube free; we call $u$ \textit{right extendable} if it has an infinite right context (or, equivalently, infinitely many right contexts). Left contexts and left extendability are defined in a dual way. 

The \textit{Thue--Morse morphism} $\theta$ is defined over $\Sigma_2^+$ by the rules $\theta(a)=ab$, $\theta(b)=ba$. The fixed points of $\theta$ are the infinite \textit{Thue-Morse word}
$$
\T = abbabaabbaababbabaababbaabbabaab\cdots
$$
and its complement, obtained from $\T$ by exchanging $a$'s and $b$'s. We refer to the factors of $\T$ as \textit{Thue-Morse factors}. The word $\T$, first introduced by Thue in \cite{Thue12} and rediscovered many times, possesses a huge number of nice properties;  we need just a few. The Thue-Morse word is overlap free, \textit{uniformly recurrent} (every Thue-Morse factor occurs in $\T$ infinitely many times with a bounded gap), and \textit{closed under reversals} ($u$ is a Thue-Morse factor iff $\overleftarrow{u}$ is).

\subsection{Uniform words and markers}
We call a word $w\in\Sigma_2^*$ \textit{uniform} if $w=c\theta(u)d$ for some $c,d\in\{a,b,\lambda\}$, $u\in\Sigma_2^*$; a uniform word with $d=\lambda$ is \textit{right aligned}. Similarly, a uniform infinite word has the form $c\theta(\u)$ for $c\in\{a,b,\lambda\}$, $u\in\Sigma_2^\infty$. Such ``almost'' $\theta$-images play a crucial role in further considerations. Note that all factors and suffixes of $\T$ are uniform. The following observation is well known.
\begin{observation}
A word $u\in\Sigma_2$ is uniform iff all occurrences of factors of the form $cc$ in $u$, where $c\in\Sigma_2$, are at positions of the same parity.
\end{observation}
Thus the word is non-uniform iff it has the factors of the form $cc$ occurring in positions of different parity. The following observation is straightforward.
\begin{observation}
A cube-free word $u\in\Sigma_2$ is non-uniform iff it contains at least one of the factors $aabaa$, $aababaa$, $bbabb$, $bbababb$.
\end{observation}
All right (resp., left) contexts of the word $ababa$ begin (resp., end) with $a$, so $ababa$ occurs in a cube-free word only as a prefix/suffix or inside the non-uniform factor $aababaa$ (same argument applies to $babab$). This allows us to view binary cube-free words as sequences of uniform factors separated by \textit{markers} $aabaa, ababa, babab$, and $bbabb$, which break uniformity.

The importance of markers for the analysis of cube-free words is demonstrated by the following theorem, proved in Section~\ref{ss:T2}.

\begin{theorem} \label{t:markers}
Every right-extendable cube-free word $u\in\Sigma_2$ has an infinite right context with finitely many markers.
\end{theorem}

\section{Proof of the Transition Property for Binary Words}

The proof of Theorem~\ref{t:uwv} consists of two stages. In the first stage we show that its result is implied by Theorem~\ref{t:markers}. In the second stage we prove Theorem~\ref{t:markers}. All words in this section are over $\Sigma_2$ if the converse is not stated explicitly.

\subsection{Reduction to Theorem~\ref{t:markers}} \label{ss:red}

\begin{lemma} \label{l:tmred}
Suppose that cube-free words $u$ and $v$ have right contexts which are Thue-Morse factors of length $2|u|$ and $2|v|$ respectively. Then there exists a word $w$ such that $uw\overleftarrow{v}$ is cube free.
\end{lemma}

\begin{proof}
Let $u_1$ and $v_1$ be the mentioned contexts of $u$ and $v$ respectively. Since $\T$ is recurrent and closed under reversals, there exists a Thue-Morse factor $w=u_1w_1\overleftarrow{v_1}$ for some $w_1\ne\lambda$. Assume to the contrary that $uw\overleftarrow{v}$ contains a cube of period $p$. Since $w$ is overlap free, such a cube cannot intersect $u$ and $\overleftarrow{v}$ simultaneously by the length argument. W.l.o.g. assume that the cube intersects $u$. Then it must contain the whole $u_1$. Hence $p$ is a period of $u_1$ and thus $p\ge|u_1|/2\ge |u|$ because $u_1$ is overlap free. Further, the overlap-freeness of $w$ means that $u$ contains at least the whole period of the cube, implying $p\le|u|$. So we have $|u|=p$ and the cube is a prefix of $uw\overleftarrow{v}$. But in this case $|u_1|\ge 2p$ and the cube is contained in the cube-free word $uu_1$, resulting in a contradiction.  
\end{proof} 

We say that a cube-free word $u$ is \textit{T-extendable} if it has a right context of the form $w\T[n..\infty]$ for some $w\in\Sigma^*$, $n\ge1$. By Lemma~\ref{l:tmred}, if the words $u$ and $\overleftarrow{v}$ are T-extendable, there is a word $w$ such that $uwv$ is cube-free. We analyze T-extendability in Lemmas~\ref{l:uniform}--\ref{l:Text}.

\begin{lemma} \label{l:uniform}
If a uniform cube-free word $u$ has a right context of length 3, or is right aligned and has a right context of length 2, then $u$ is T-extendable.
\end{lemma}

\begin{proof}
We assume $|u|\ge 5$; otherwise, $u$ is a factor of $\T$ and there is nothing to prove. Consider two cases.\\
\textit{Case 1}: $u=c\theta(v)=cv_1\ldots v_n$, $v_i \in \{ab,ba\}$ for $i\in[1..n]$. W.l.o.g., $v_n=ab$.\\
\textit{Claim 1.} At least one of the words $ua$, $ubab$ is cube free.\\
Assume that $ua$ is not cube free and thus has a suffix $w^3$ of period $p$. Then $u$ has the suffix $u'$ of length $3p-1$ and period $p$. If $u'\ne u$, by the  cube-freeness of $u$ one has $u[|u|{-}3p{+}1]\ne u[|u|{-}2p{+}1]$, $u[|u|{-}3p{+}2]= u[|u|{-}2p{+}2]$. Note that $u[|u|{-}2p{+}1..|u|{-}2p{+}2]=v_{n-p+1}\in\{ab,ba\}$. Then $u[|u|{-}3p{+}1]=u[|u|{-}3p{+}2]$, so these two letters do not form a block $v_i$; then $p$ must be odd. Consider the suffix $xx=v_{n-p+1}\cdots v_n$ of length $2p$ of $u'$. Both prefix and suffix of $x$ of length $p-1$ are concatenations of blocks $ab,ba$. Hence $x$ consists of alternating letters. Then $x[1]=x[p]$. On the other hand, $v_{n-\lfloor p/2\rfloor}=x[p]x[1]$. This contradiction proves that $u'=u$ and then $|u|=3p-1$. Further, since $ua$ and $ubb$ end with cubes, the only length-2 right context of $u$ is $ba$, so $uba=c\theta(vb)$ is cube-free. If $ubab$ ends with a cube of period $p'$, we repeat the above argument for $uba$ to obtain $|uba|=3p'-1$. Hence $3p-1=3p'-3$, which is impossible since the periods are integers. So $ubab$ is cube free and Claim 1 holds.\\
\textit{Claim 2.} At least one of the words $uaa$, $ubabb$ is cube free.\\
Assume that $ua$ is cube free. If $uaa$ has a suffix $w^3$ of period $p$, then $w$ ends with $aa$ (recall that $u$ ends with $ab$). On the other hand, the leftmost $w$ in the suffix $w^3$ of $uaa$ ends with $v_{n-p+1}\in\{ab,ba\}$. Thus $uaa$ cannot have a cube as a suffix and hence is cube free. The same argument works for $ubabb$ if $ubab$ is cube free. The reference to Claim 1 concludes the proof.

\smallskip
Assuming that $uaa$ is cube free, we show that the word 
$$
\v=u\T[6..\infty]= cv_1\cdots v_{n-1}ab\,aabbaababba\cdots
$$
is cube free and then $u$ is T-extendable. Depending on $v_{n-1}$, $\v$ has the non-uniform factor $aabaa$ or $aababaa$, and this factor, denoted by $x$, has a unique occurrence in $\v$ because $u$ and $\T$ are uniform. Note that $ab\T[6..\infty]=\T[4..\infty]$, and both words $a\T[4..\infty]$ and $aab\T[4..\infty]$ are cube free. So if a cube $w^3$ is a factor of $\v$, then Claim~2 implies that $x$ is a factor of $w^3$. If $x$ is a factor of $w^2$, then $x$ occurs in $w^3$ at least twice, which is not the case. So $x=w'ww''$, where $w'$ and $w''$ are nonempty suffix and nonempty prefix of $w$ respectively. Then $w=aba$ if $x=aabaa$, and $w=ababa$ if $x=aababaa$. In both cases a direct check shows that $w^3$ is not a factor of $\v$. So we proved that $\v$ is cube free. Assuming that $ubabb$ is cube free, we use the same argument for another suffix of $\T$:
$$
\v=u\T[20..\infty]= cv_1\cdots v_{n-1}ab\,babbaabbaba\cdots
$$
(here $\v$ contains a unique occurrence of $bbabb$). Note that if $v_n=ba$, then we can take $\T[20..\infty]$ (resp., $\T[6..\infty]$) as the extension of $u$ if $ubb$ (resp., $uabaa$) is cube free.

\smallskip\noindent\textit{Case 2}. $u=c\theta(v)a=cv_1\ldots v_na, v_i \in \{ab,ba\}$  for $i\in[1..n]$.\\
\textit{Case 2.1}: $v_n=ab$. The word $u\T[7..\infty]=cv_1\cdots v_n\T[6..\infty]$ is cube free as in Case~1.\\
\textit{Case 2.2}: $v_n=ba$. Since $ua$ ends with $a^3$, $ub$ is cube free, right aligned and has a right context of length 2. Then $ub$ is T-extendable by Case~1, and so is $u$.  
\end{proof}

\begin{lemma} \label{l:uniform2}
If a cube-free word $u$ has a uniform right context $w$ such that $|w|\ge 2|u|+3$ and $w$ has no prefix $ababa$ or $babab$, then $u$ is T-extendable. 
\end{lemma}

\begin{proof}
Let $\hat w$ be the right aligned prefix of $w$ of length $2|u|$ or $2|u|{+}1$. We will prove that $u\hat w$ is $T$-extendable, which implies the result immediately. Suppose $u\hat w$ is non-uniform (otherwise, it is T-extendable by Lemma~\ref{l:uniform}). Then it contains markers, and all of them begin in $u$ by the conditions on $w$. Let $z$ be the rightmost marker in $u\hat w$. W.l.o.g. the first letter of $z$ is $a$ and we can write $u=u'au''$, where $au''w$ begins with this distinguished occurrence of $z$. The pair $(u'a,u''w)$ satisfies all conditions of the lemma, so for the rest of the proof we rename $u'a$ as $u$ and $u''\hat w$ as $\hat w$. This renaming retains the value of the word $u\hat w$ we analyze; still, $\hat w$ is right aligned and $|\hat w|\ge 2|u|$. Since $\hat w$ has a right context of length at least 2, it is T-extendable by Lemma~\ref{l:uniform}, and, moreover, there is a suffix $\v$ of $\T$ such that $\hat w\v$ is cube-free. If $\hat w$ is a factor of $\T$, we choose $\v$ such that $\hat w\v$ is a suffix of $\T$; otherwise, $\v$ is chosen as in the proof of Lemma~\ref{l:uniform}, Case~1.

Assume to the contrary that $u\hat w\v$ contains a cube $x^3$; it starts in $u$ and ends in $\v$, thus containing the distinguished occurrence of $z$:\\
\centerline{\includegraphics[scale=0.9, trim = 25mm 235mm 78mm 43mm, clip]{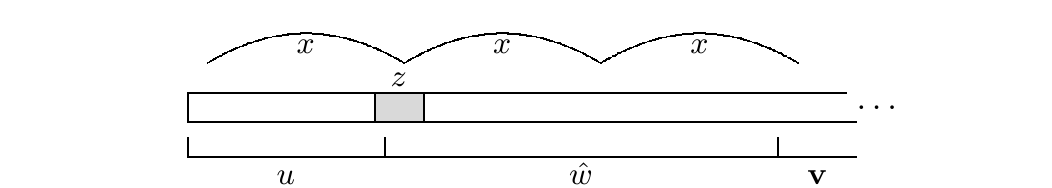}}
If $z$ occurs in $x^3$ only once, we have one of two cases, up to symmetry:
$$
\underbrace{\cdots ab\pmb{a}}_u \underbrace{\pmb{aba\,a}ba\cdots}_{\hat w\v}\qquad\text{or}\qquad
\underbrace{\cdots \pmb{a}}_u \underbrace{\pmb{b\,ab\,a}b\cdots}_{\hat w\v}
$$
In the first case the condition $|\hat w|\ge 2|u|$ is violated; the second case contradicts the choice of $\v$ (here $\hat w$ is a factor of $\T$). Therefore, $z$ must have two occurrences in $x^3$ at distance $|x|$. If $z$ occurs in $x^3$ to the left of the distinguished occurrence, then $|x|<|u|$ and $x^3$ cannot end in $\v$. Otherwise, $x^3$ contains exactly two occurrences of $z$: the distinguished one and another one on the border of $\hat w$ and $\v$. Then $x$ does not contain $z$, implying $|x|\le |u|+3$. On the other hand, $|x|\ge|\hat w|-3$ as the distance between the occurrences of $z$. Hence $|u|\le 6$ and $|x|\le 9$. This leaves, up to symmetry, the following options for $x^3$:
$$
\begin{array}{llll}
|x|=5:&aba\pmb{a\,babaabab\,a}ab\cdots &|x|=9:&aabba\pmb{a\,baba}abba\pmb{abab\,a}abbaabab\cdots\\
|x|=7:&aabb\pmb{a\,abaa}bb\pmb{aaba\,a}bbaab\cdots
&&ababb\pmb{a\,abaa}babb\pmb{aaba\,a}babbaaba\cdots\end{array}
$$
The factor between the marginal letters of markers contains $\hat w$; in first three cases, this factor occurs in $\T$ but $\hat w\v$ is not a suffix of $\T$, contradicting the choice of $\v$. In the last case $\hat w$ is not a factor of $\T$, so $\v$ is chosen as in the proof of Lemma~\ref{l:uniform}; hence the marker on the border of $\hat w$ and $\v$ must be followed by $bb$, not $ba$. This contradiction finishes the proof. 
\end{proof}

Some right-extendable words have no long uniform right contexts, as Fig.~\ref{f:tree} shows. However, a weaker property is enough for our purposes.
\begin{figure}[!htb]
\centerline{\includegraphics[scale=1, trim = 10mm 259mm 75mm 19mm, clip]{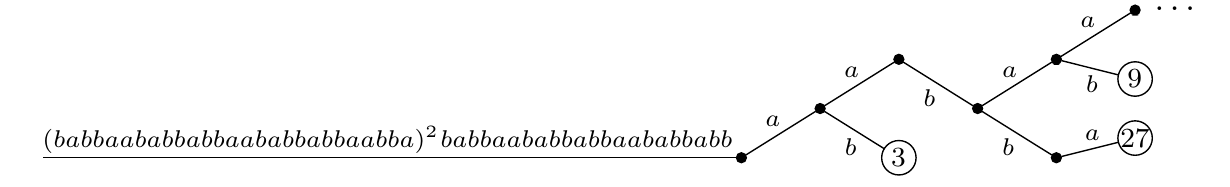}}
\caption{A right-extendable word of length 76 having no long uniform right contexts:  all its infinite right contexts begin with the marker $aabaa$.}
\label{f:tree}
\vspace*{-3mm}
\end{figure}

\begin{lemma} \label{l:Text}
Every cube-free word having an infinite right context with finitely many markers is T-extendable.
\end{lemma}

\begin{proof}
Let $u$ be the word and $\w$ be its context from the conditions of the lemma. The finiteness of the number of markers allows us to write $\w=w_1\v$, where $\v$ is uniform. Then $uw_1$ has an infinite uniform right context, and hence is T-extendable by Lemma~\ref{l:uniform2}. Then $u$ is T-extendable as well. 
\end{proof}

Thus if Theorem~\ref{t:markers} holds, then Lemmas~\ref{l:Text} and~\ref{l:tmred} imply Theorem~\ref{t:uwv}.

\subsection{Proof of Theorem~\ref{t:markers}} \label{ss:T2}

We prove Theorem~\ref{t:markers} by reductio ad absurdum; to obtain a contradiction, we use the following lemma on cube-free words over an arbitrary alphabet. (For Theorem~\ref{t:markers}, bounding $k$ in Lemma~\ref{l:log} by \textit{any} function of $n$ would be sufficient; however, better bounds can be useful for the algorithmic applications.)

\begin{lemma} \label{l:log}
Let $u$ be a cube-free word of length $n$ over an \emph{arbitrary} fixed alphabet and let $u$ have a length-$k$ right context $w$ with the following property: for each $i=1,\ldots,k$, there exists an integer $p_i\ge 2$ such that the suffix of length $3p_i-2$ of the word $u{\cdot}w[1..i]$ has period $p_i$ and, moreover, $p_i\ne p_{i+1}$. Then $k=O(\log n)$; more precisely, $k\le \max\{1, 8.13\log n-15.64\}$. 
\end{lemma}

\begin{proof}
In the proof we can assume $k\ge 2$. Let $1\le i<j \le k$, $p=p_i$, $q=p_j$, $l=j-i$, and let $v$ be the intersection of the periodic suffixes of $u{\cdot}w[1..i]$ and $u{\cdot}w[1..j]$ (see Fig.~\ref{f:suffixes}a,b). If $|v|\ge p+q-\gcd(p,q)$, then $v$ has the period $\gcd(p,q)$ by the Fine--Wilf property (Lemma~\ref{fw}). If $p\ne q$, this means that the root of the longer periodic suffix is an integer power of a shorter word; thus $uw$ contains a cube, which is impossible. If $p=q$, then we are in the situation shown in Fig.~\ref{f:suffixes}b, and the union of two suffixes has period $p$ and the length $3p-2+l$. In this case, $l\ge 2$ by conditions of the lemma, so we again obtain a cube. Thus we conclude that 
\begin{equation} \label{e:fw}
|v|\le p+q-\gcd(p,q)-1.
\end{equation}

\begin{figure}[!htb]
\centerline{\includegraphics[scale =0.83, trim = 38mm 234mm 10mm 40mm, clip]{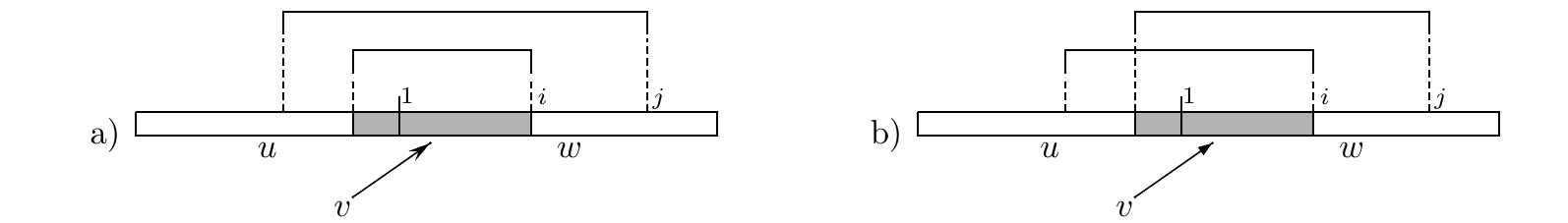}}
\caption{The mutual location of periodic factors in the word $uw$ (Lemma~\ref{l:log}).}
\label{f:suffixes}
\end{figure}
The case in Fig.~\ref{f:suffixes}b  corresponds to $|v|=3q-2-l$. Comparing this condition to \eqref{e:fw}, we get $q\le \frac{p+l}2$ (but $q=p$ only if $l\ge 2p-1$ and $q=p/2$ only if $l\ge p/2-1$). Similarly, the case in Fig.~\ref{f:suffixes}a corresponds to $|v|=3p-2$ and we get $q>2p$ from \eqref{e:fw}. Thus, all possible values of $q$ are outside the red area in Fig.~\ref{f:pql}.

\begin{figure}[!htb]
\centerline{\includegraphics[scale =0.95, trim = 5mm 227mm 95mm 10.5mm, clip]{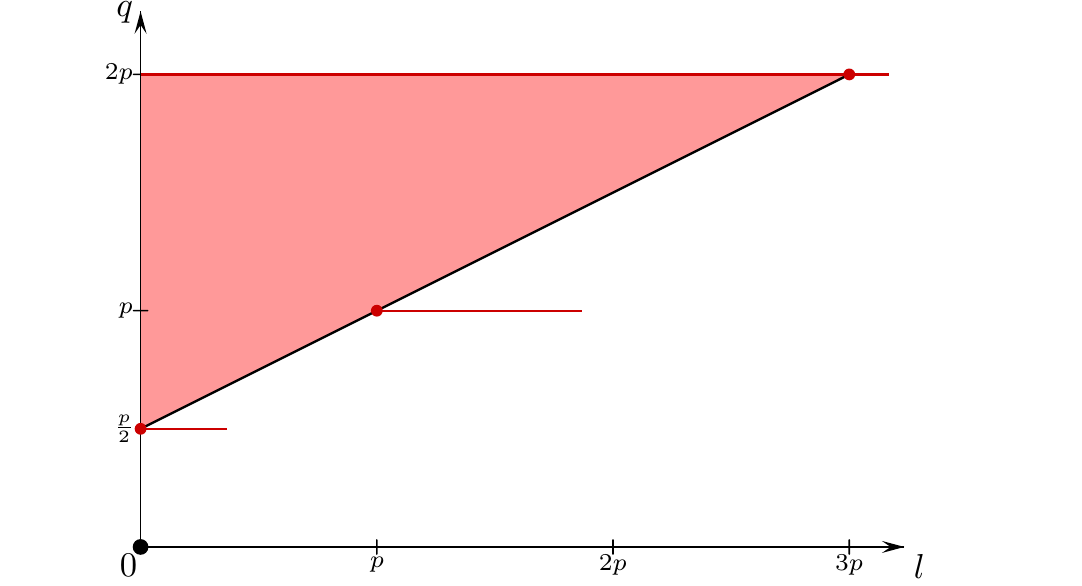}}
\caption{The restrictions on periods of periodic factors in the word $uw$ (Lemma~\ref{l:log}).}
\label{f:pql}
\end{figure}
Now we estimate how many elements of the sequence $\{p_1,\ldots,p_k\}$ can belong to the range $[p..2p]$ for some fixed $p\ge 2$. This is an analog of \cite[Lemmas 4,5]{PeSh15} and \cite[Lemma 9]{PeSh17}. Let $i_0<i_1<\cdots <i_s$ be the list of all positions such that the periodic suffix of $u{\cdot}w[1..i_j]$ has the period from the range $[p..2p]$; let $q_0,\ldots,q_s$ denote these periods. Then Fig.~\ref{f:pql} gives us the lower bound for the distance $l_j=i_{j+1}-i_j$ between consecutive positions from the list:
\begin{equation} \label{e:dist}
l_j\ge 2q_{j+1}-q_j;\ l_j\ge 2q_{j+1}-1 \text{ if }q_{j+1}=q_j;\ l_j\ge q_{j+1}-1 \text{ if }q_{j+1}=q_j/2. 
\end{equation}
The densest packing of the numbers $i_j$, satisfying the restrictions \eqref{e:dist}, is achieved for $q_0=2p-1$, $q_1=q_3=q_5=\cdots=p$, $q_2=q_4=q_6=\cdots=p+1$: one can take $i_0=1$, $i_1=2$, and $i_{2j}=i_{2j-1}+p+2$, $i_{2j+1}=i_{2j}+p-1$ for all subsequent positions. Since $i_s\le k$, we have $\big\lceil \frac {s-1}2\big\rceil\cdot(p+2)+\big\lfloor \frac {s-1}2\big\rfloor\cdot(p-1)+2\le k$, implying the upper bound for the number of periods from the range $[p..2p]$:
\begin{equation} \label{e:range}
s{+}1\le \frac {2(k-2)}{2p+1}+2.
\end{equation}
Since $3p-2\le |uw| = n+k$, the maximum possible value of $p$ is $\big\lfloor \frac {n+k+2}3\big\rfloor$. We partition all possible periods into $r$ ranges of the form $[p..2p]$:
$$
[2..4], [5..10], [11..22],\ldots,\big[3\cdot 2^{r-1}{-}1..\big\lfloor \tfrac {n+k+2}3\big\rfloor\big].
$$
The number of ranges thus satisfies $r\le \log \frac {n+k+2}9 +1$. The sum of the upper bounds \eqref{e:range} for all ranges is at least $k$; observing that the number $2p{+}1$ in \eqref{e:range} is the first period from the range next to $[p..2p]$, we can write
\begin{equation} \label{e:k1}
k\le 2(k-2){\cdot}\sum_{i=1}^r \frac 1{3\cdot 2^i{-}1} + 2r.
\end{equation}
The sum in \eqref{e:k1} is bounded by $\frac 15+ \frac 1{11}+\frac 1{23} + \frac 1{47}\cdot\sum_{i=0}^\infty \frac 1{2^i} < 0.377;$
substituting this value and the upper bound for $r$, we get
\begin{equation} \label{e:k2}
0.246(k-2)\le 2 \log \tfrac {n+k+2}9.
\end{equation}
For $k\ge n-1$, \eqref{e:k2} implies $0.246(k-2)\le 2 \log \tfrac {2k+3}9$, but this inequality fails for $k\ge 2$. So $k\le n-2$ and we replace \eqref{e:k2} with the inequality $0.246(k-2)\le 2 \log \tfrac {2n}9$; it gives, after arithmetic transformations, the required bound on $k$. 
\end{proof}

\begin{proof}[Proof of Theorem~\ref{t:markers}]
For the sake of contradiction, assume that all infinite right contexts of some right-extendable cube-free word $u$ contain infinitely many markers. W.l.o.g. we can assume that $u$ ends with a marker (if not, choose a prefix $v$ of an infinite right context of $u$ such that $uv$ ends with a marker, and replace $u$ with the word $uv$ having the same property of right contexts). Let $z$ be the marker which is a suffix of $u$. For example, if $u$ is the word written in the ``trunk'' of the tree in Fig.~\ref{f:tree}, then $z=bbabb$.

By our assumption, $u$ has no infinite uniform right contexts. Thus $u$ has finitely many uniform right contexts (in Fig.~\ref{f:tree} such contexts are $\lambda$, $a$, $aa$, $aab$, $aaba$, and $aabb$). Other uniform words, being appended to $u$, produce cubes; in Fig.~\ref{f:tree}, appending a word beginning with $ab$ (resp., $aabab$, $aabba$) gives the cube $(bab)^3$ (resp., $(babbaabab)^3$, $(babbaababbabbaababbabbaabba)^3$). All these cubes contain markers and illustrate three types of cubes with respect to the occurrences of markers ($x$ below denotes the root of the cube):
\begin{itemize}
\item \textit{mini}: $x^2$ contains no markers, $x^3$ contains a marker (example: $x=bab$);
\item \textit{midi}: $x$ contains no markers, $x^2$ contains a marker (example: $x=babbaabab$);
\item \textit{maxi}: $x$ contains markers (example: $x=babbaababbabbaababbabbaabba$);
\end{itemize}
Note that mini cubes are exactly those having the root $x\in\{ab,ba,aba,bab\}$. Further, the intersection of two markers in a cube-free word is either empty or one-letter; this fact implies that each midi cube contains exactly two markers.

We call $w$ a \textit{semi-context} of $u$ if $uw$ ends with a cube but $u{\cdot} w[1..|w|{-}1]$ is cube free. Next we show the following fact: 
\begin{itemize}
\item[$(*)$] The word $u$ has two distinct uniform semi-contexts $w_1$ and $w_2$ such that $uw_1$ and $uw_2$ end with midi or maxi cubes.
\end{itemize}
To prove $(*)$ we need a case analysis. W.l.o.g., $z$ begins with $a$. Let $\w$ be an infinite right context of $u$.

\smallskip\noindent
Case 1: $z=aabaa$. We have $u=\cdots baabaa$, so $\w$ cannot begin with $a$ or $baa$ because $u\w$ is cube-free. So $\w$ begins with $bba$, $babba$, or $baba$. In the first case, some prefixes of words
\begin{align} \label{e:t1}
u\T[2..\infty]&=\cdots baabaa\,bba\,baababba\cdots,\\
u\T[22..\infty]&=\cdots baabaa\,bba\,abbabaab\cdots
\end{align}
must end with cubes. The longest common prefix $ubba$ of these words is cube free as a prefix of $u\w$, so these cubes are different, contain the marker $z$, and are not mini. Hence we can take some prefixes of $\T[2..\infty]$ and $\T[22..\infty]$ as the semi-contexts required in $(*)$. If $\w$ begins with $babba$, the same result is obtained with prefixes of the words
\begin{align} \label{e:t2}
u\T[12..\infty]&=\cdots baabaa\,babba\,baabbaab\cdots,\\
u\T[20..\infty]&=\cdots baabaa\,babba\,abbabaab\cdots
\end{align}
Finally, if $\w$ begins with $baba$, we can take one word from each pair (say, $u\T[2..\infty]$ and $u\T[12..\infty]$). Their longest common prefix is the cube-free word $ub$, so the cubes given by the corresponding semi-contexts are distinct.

\smallskip\noindent
Case 2: $z=ababa$. Here $\w=ab\cdots$.  Taking the pair of words
\begin{align} \label{e:t3}
u\T[7..\infty]&=\cdots ababa\,ab\,baababba\cdots,\\
u\T[19..\infty]&=\cdots ababa\,ab\,abbaabba\cdots,
\end{align}
we achieve the same result as in Case 1: some prefixes of these words end with midi or maxi cubes, and these cubes are distinct because the common prefix $uab$ of the presented words is cube-free. Thus, $(*)$ is proved.

\medskip
Now take the semi-contexts $w_1$, $w_2$ given by $(*)$ such that $uw_1$ and $uw_2$ end with cubes $x_1^3$ of period $p_1$ and $x_2^3$ of period $p_2$ respectively. Let $w$ be the longest common prefix of $w_1$ and $w_2$; w.l.o.g., $w_1=waw_1'$, $w_2=wbw_2'$. In both $x_1^3$ and $x_2^3$, the suffix $z$ of $u$ is the rightmost marker and hence matches an earlier occurrence of the same marker in $u$. These occurrences are different, because $z$ is followed by $wa$ in $x_1^3$ and by $wb$ in $x_2^3$. In particular, $p_1\ne p_2$. W.l.o.g., $p_1>p_2$; then $x_1$ contains $z$ and so $x_1^3$ is maxi (see Fig.~\ref{f:w1w2}).
\begin{figure}[!htb]
\vspace*{-3mm}
\centerline{\includegraphics[scale =1, trim = 10mm 232mm 38mm 41mm, clip]{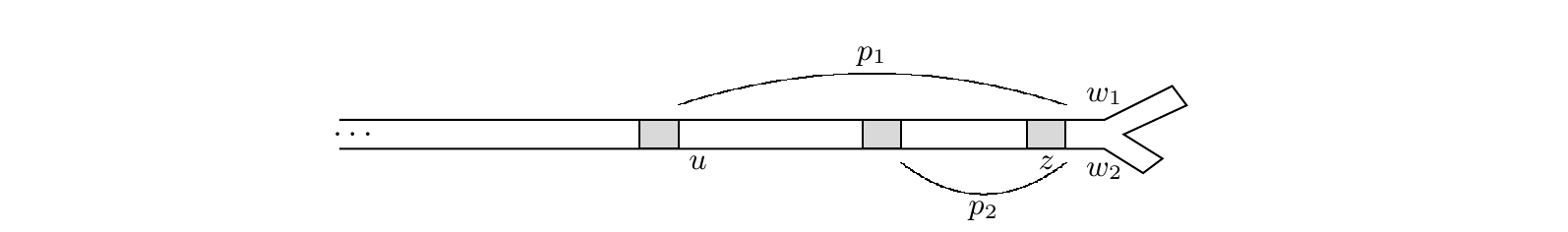}}
\caption{Semi-contexts $w_1$ and $w_2$ of the word $u$: periods of cubes and corresponding markers. Three grey factors are markers equal to $z$, other markers are not shown.}
\label{f:w1w2}
\end{figure}

Let $z_1=z, z_2,\ldots, z_m$ be all markers in $u$, right to left. We factorize $u$ as $u=y_m\cdots y_2y_1$, where $y_i$ begins with the first letter after $z_{i+1}$ ($y_m$ is a prefix of $u$) and ends with the last letter of $z_i$ (even if $z_{i+1}$ and $z_i$ overlap); see Fig.~\ref{f:factorize} for the example. Assume that $z_{j+1}$ matches $z_1$ in the maxi cube $x_1^3$ (note that $j\ge 2$, because a marker with a smaller number matches $z_1$ in $x_2^3$; $j=3$ in Fig.~\ref{f:factorize}). Then $z_{j-1},\ldots, z_1$ are in the rightmost occurrence of $x_1$, and $z_j$ is either also in this occurrence or on the border between the middle and the rightmost occurrences (in Fig.~\ref{f:factorize}, the latter case is shown). Depending on this, $x_1^3$ contains either $3j$ or $3j{-}1$ markers. Further, we see that $w_1$ is a prefix of $y_j$, $y_1=y_{j+1},\ldots,y_{2j-2}=y_{3j-2}$.
\begin{figure}[!htb]
\centerline{\includegraphics[scale =1.1, trim = 15mm 235mm 38mm 41mm,  clip]{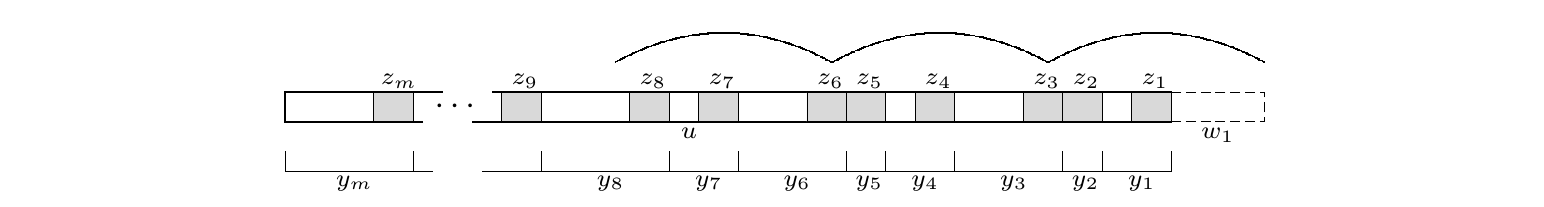}}
\caption{Marker-based factorization of the word $u$. Markers are grey, arcs indicate the cube after appending $w_1$ to $u$.}
\label{f:factorize}
\end{figure}

Let us extend $u$ to the right by a context $y_0$ such that $uy_0$ is right extendable, $y_0$ ends with a marker $z_0$, and all proper prefixes of $y_0$ are uniform. Applying all the above argument to $uy_0$ and its factorization $y_m\cdots y_1y_0$, we get another maxi cube (say, $x_0^3$) and the corresponding set of equalities between $y_i$'s. Note that $y_0\ne y_j$: as was mentioned in the previous paragraph, $y_j$ has the prefix $w_1$, while $y_0$ cannot have this prefix because $uw_1$ contains a cube.

\medskip
Let us iterate the procedure of appending a context $k$ times, getting a right-extendable word $uy=y_m\cdots y_1y_0\cdots y_{1-k}$ as the result (according to our assumption on $u$, the number $k$ can be arbitrarily big). Now consider the finite alphabet $\Gamma=\{y_m,\ldots, y_1,y_0,\ldots, y_{1-k}\}$ and let $U=y_m\cdots y_2y_1$, $Y=y_0\cdots y_{1-k}$ be words over $\Gamma$. They are cube free and $Y$ is a length-$k$ right context of $U$. Each word $U{\cdot}Y[1..i]$ ends with a suffix having some period $p_i$ and length $3p_i-2$ or $3p_i-1$. In addition, $p_i\ne p_{i+1}$, because $y_{-i}\ne y_{p_i-i}$. So all conditions of Lemma~\ref{l:log} are satisfied, and we apply it to get an upper bound on $k$. The existence of this bound contradicts our assumption that all infinite right contexts of $u$ have infinitely many markers. The theorem is proved.
\end{proof}

\section{Transition Property for Big Alphabets}

Here we extend the results of the previous section to arbitrary finite alphabets.

\begin{theorem} \label{t:uwv3}
For every $d\ge 3$ and every pair $(u,v)$ of $d$-ary cube-free words such that $u$ is right extendable and $v$ is left extendable, there exists a $d$-ary word $w$ such that $uwv$ is cube free.
\end{theorem}

As in the binary case, we use an auxiliary theorem about the existence of a context with finitely many markers (but the markers are different now).

\begin{theorem} \label{t:markers2}
Let $d\ge 3$. Every right-extendable cube-free word $u\in\Sigma_d^*$ has an infinite right context with finitely many occurrences of all letters except for $a$ and $b$.
\end{theorem}

\begin{proof}
We follow the main idea of the proof of Theorem~\ref{t:markers} and use the same notation. The difference, which actually simplifies the argument, is that the role of markers is now played by the \textit{$c$-letters} $c_1,\ldots, c_{d-2}$. Aiming at a contradiction, assume that all infinite right contexts of some right-extendable cube-free word $u$ contain infinitely many $c$-letters. W.l.o.g. we can assume that $u$ ends with a $c$-letter; we denote this letter by $z$. As in the proof of Theorem~\ref{t:markers}, we let $z_1=z, z_2,\ldots, z_m$ be all markers in $u$, right to left, and factorize $u$ as $u=y_m\cdots y_2y_1$, where $y_i$ begins with the first letter after $z_{i+1}$ and ends with the last letter of $z_i$.

By our assumption, $u$ contains finitely many contexts from $\{a,b\}^*$; then it has two semi-contexts $w_1=waw_1', w_2=wbw_2'\in\{a,b\}^+$ (each of the words $w,w_1,w_2$ may be empty). Let $x_1^3, x_2^3$ be suffixes of $uw_1$ and $uw_2$ respectively, and w.l.o.g. $|x_1|=p_1>p_2=|x_2|$. The suffix $zw_1$ of $x_1$ matches some earlier occurrence of $zw_1$ in $u$; same for the suffix $zw_2$ of $x_2$. As in the proof of Theorem~\ref{t:markers} we see that $z_{j+1}=z_1$ for some $j\ge 2$, $w_1$ is a prefix of $y_j$, and the equalities $y_1=y_{j+1},\ldots, y_{2j-1}=y_{3j-1}$ hold.

Next we extend $u$ to the right by a context $y_0$ such that $uy_0$ is right extendable, $y_0$ ends with a marker $z_0$, and all proper prefixes of $y_0$ are over $\{a,b\}$. Applying all the above argument to $uy_0$ and its factorization $y_m\cdots y_1y_0$, we get another cube $x_0^3$ and the corresponding set of equalities between $y_i$'s. Again, $y_0\ne y_j$, since $y_j$ has the prefix $w_1$, while $y_0$ has not. After iterating the procedure of appending a context $k$ times, we obtain a right-extendable word $uy=y_m\cdots y_1y_0\cdots y_{1-k}$ and consider the words $U=y_m\cdots y_2y_1$, $Y=y_0\cdots y_{1-k}$ over the alphabet $\Gamma=\{y_m,\ldots, y_1,y_0,\ldots, y_{1-k}\}$. They are cube free and $Y$ is a length-$k$ right context of $U$. Each word $UY[1..i]$ ends with a suffix having some period $p_i$ and length $3p_i-1$. In addition, $p_i\ne p_{i+1}$, because $y_{-i}\ne y_{p_i-i}$. So we can apply Lemma~\ref{l:log} to get an upper bound on $k$. The existence of this bound contradicts our assumption that all infinite right contexts of $u$ have infinitely many $c$-letters. Hence $u$ has an infinite context with finitely many $c$-letters, as required.
\end{proof}

\begin{proof}[Proof of Theorem~\ref{t:uwv3}]
By Theorem~\ref{t:markers2}, the word $u\in\Sigma_d^*$ has an infinite right context with finitely many $c$-letters. First we note that we can choose such a context containing a $c$-letter (if a context $\w$ is over $\Sigma_2$, one can get another context of $u$ replacing, say, the letter $\w[|u|]$ with $c_1$). So we can write this context $\w$ as $x\u_1$, where $x$ ends with a $c$-letter and $\u_1\in\Sigma_2^\infty$. Let $u_1$ be the prefix of $\u_1$ of length $\lceil|ux|/2\rceil$. In the same way, we take a right context $\overleftarrow{y}\overleftarrow{\v_1}$ of $\overleftarrow{v}$ and the prefix $\overleftarrow{v_1}\in\{a,b\}^*$ of $\overleftarrow{\v_1}$ of length $\lceil|yv|/2\rceil$. Then the binary words $u_1,v_1$ are cube free, $u_1$ is right extendable, and $v_1$ is left extendable. Applying Theorem~\ref{t:uwv}, we take a binary transition word $w_1$ such that $u_1w_1v_1$ is cube free. Then $s=uxu_1w_1v_1yv$ is cube free. Indeed, $x$ ends with a $c$-letter, $y$ begins with a $c$-letter, and these $c$-letters are separated by a cube-free word over $\Sigma_2$. Hence a cube in $s$, if any, must contain one of these $c$-letters. But the lower bounds on $|u_1|$ and $|v_1|$ imply that this $c$-letter cannot match another $c$-letter to produce a cube (recall that $uxu_1$ and $v_1yv$ are cube free). Thus $s$ is cube-free and we obtain a transition word $xu_1w_1v_1y$ for the pair $(u,v)$.
\end{proof}

\section{Solving the Restivo-Salemi Problems and Future Work}

To give the solutions to the Restivo--Salemi Problems~2,~4,~and~5, recall the solution to Problem~1a \cite{CuSh03}: a $d$-ary $\alpha$-power-free word $u$ is right extendable iff it has a right context of length $f_{\alpha,d}(|u|)$ for some computable function $f_{\alpha,d}$. Algorithm~\ref{alg:4} below solves Problem~4.

\begin{algorithm*}
\caption{: Deciding the existence of a transition word for cube-free words $u,v\in\Sigma_d$}
\label{alg:4}
\vspace*{-2mm}
\begin{itemize}
\item For all words $w\in\Sigma_d$ such that $|w|\le f_{3,d}(|u|)$ check whether $w$ is a right context of $u$ 
\item If a context $w$ with the suffix $v$ is found, return ``yes''
\item If no context of length $f_{3,d}(|u|)$ is found, return ``no''
\item Else \hfill  $\triangleright\ u$ is right extendable 
\begin{itemize}
\item For all words $w\in\Sigma_d$ such that $|w|\le f_{3,d}(|v|)$ check whether $w$ is a left context of $v$
\item If a context $w$ with the prefix $u$ is found, return ``yes''
\item If no context of length $f_{3,d}(|v|)$ is found, return ``no''
\item Else return ``yes'' \hfill  $\triangleright\ v$ is left extendable; apply Theorem~\ref{t:uwv} or \ref{t:uwv3}
\end{itemize}
\end{itemize}
\vspace*{-2mm}
\end{algorithm*}

The natural next step is to find an \textit{efficient} algorithm for Problem~4. The function $f_{3,d}(n)$ is sublinear, but the search space is still of size $2^{n^{\Omega(1)}}$. The possible way to a polynomial-time solution is to strengthen the connection with Lemma~\ref{l:log} to show that it is sufficient to process the contexts of length $O(\log n)$, where $n=\max\{|u|,|v|\}$.

\smallskip
For Problem 2, the first step is the reduction to the binary case. Let $u\in\Sigma_d$, $d\ge 3$, be a right-extendable cube-free word; we write $u=u'cu''$, where $c$ is the rightmost $c$-letter in $u$. We check all cube-free words $w\in\Sigma_2^*$ such that $u''w$ is right extendable and $|u''w|=\lceil|u|/2\rceil$ for being right contexts of $u$. If $w$ is a right context of $u$, then any (binary) right context of $u''w$ is a right context of $uw$, so the problem is reduced to binary words. If no word $w$ suits, we take the shortest right context of $u$ of the form $vc$, where $v\in\Sigma_2^*$, $c\in\{c_1,\ldots,c_{d-2}\}$ such that the word $u_1=uvc$ is right extendable; such a context can be found in finite time because $|v|<|u|/2$. Then we replace $u$ by $u_1$ and repeat the search of long binary right contexts. By Theorem~\ref{t:markers2}, we will succeed after a finite number of iterations, and Lemma~\ref{l:log} gives the upper bound on the maximum number $k$ of iterations depending on $|u|$. Thus we end this step getting a word $y\hat u$ such that $uy\hat u$ is cube free, $\hat u\in\Sigma_2^*$, and all binary right contexts of $\hat u$ are right contexts of $uy\hat u$. If $u$ is binary, we skip this step setting $\hat u=u$.

On the second step we further reduce the problem to uniform words. We act as in the first step, using Theorem~\ref{t:markers} and Lemma~\ref{l:uniform2}. Namely, we check for uniform contexts and if $\hat u$ has no uniform context $w$ of length $2|\hat u|+3$ such that $\hat uw$ is right extendable and $w$ has no prefix $ababa/babab$, we append the shortest context $v$ ending with a marker, repeating the search for $\hat u_1=\hat uv$. Theorem~\ref{t:markers} guarantees that we will find the required uniform context in at most $k$ iterations, where $k$ is as in Lemma~\ref{l:log}. Thus at this step we build a right context $\hat y\hat w$ of $\hat u$ such that $\hat u\hat y\hat w$ is right extendable, $|\hat{w}|\ge 2|\hat u\hat y|+3$ and $\hat{w}$ is uniform and has no prefix $ababa/babab$.

Finally we choose, as described in Lemma~\ref{l:uniform2}, a suffix $\T[r..\infty]$ of $\T$ which is a right context of $\hat w$: if $\hat w=\T[i..j]$ for some $i,j$, then we take $r=j+1$, otherwise the choice is performed according to Case 1 in the proof of Lemma~\ref{l:uniform}. Now Lemma~\ref{l:uniform2} guarantees that $\hat u\hat y\hat w\T[r..\infty]$ is cube free. Thus the infinite right context of the original word $u$ is given by the finite word $Y=y\hat u\hat y\hat w$ and the number $r$. The above description is summarized below as Algorithm~\ref{alg:2}.

\begin{algorithm*}
\caption{: Finding an infinite right context of a right-extendable cube-free word $u\in\Sigma_d$}
\label{alg:2}
\vspace*{-2mm}
\begin{itemize}
\item $U\gets u$, $Y\gets \lambda$
\item If $U$ is not binary\hfill  $\triangleright$ first step 
\begin{itemize}
\item While $U$ has no long binary right context
\begin{itemize}
\item Find the shortest right context $vc$ ending with a $c$-letter
\item $U\gets Uvc$, $Y\gets Yvc$
\end{itemize}
\item $\hat u\gets \text{ long binary right context of } U$, $Y\gets Y\hat u$ 
\end{itemize}
\item Else $\hat u=u$
\item While $\hat u$ has no long uniform right context without prefix $ababa/babab$ \hfill  $\triangleright$ second step
\begin{itemize}
\item Find the shortest right context $v$ such that $\hat uv$ ends with a marker
\item $\hat u\gets \hat uv$, $Y\gets Yv$
\end{itemize}
\item $\hat w\gets \text{ long uniform right context of } \hat u$ without prefix $ababa/babab$, $Y\gets Y\hat w$ 
\item Find $r$ such that $\T[r..\infty]$ is a right context of $\hat w$\hfill  $\triangleright$ final step
\item return $Y,r$
\end{itemize}
\end{algorithm*}

Again, the natural direction of the future work is to make Algorithm~\ref{alg:2} efficient. 

Finally we approach Problem~5. We first run Algorithm~\ref{alg:4}, which can provide us with an example of a transition word if $u$ or $\overleftarrow{v}$ is not right extendable. If both $u,\overleftarrow{v}$ are right extendable, we run for each of them Algorithm~\ref{alg:2}, getting $Y_1,Y_2,r_1,r_2$ such that $uY_1\T[r_1..\infty]$ and $\overleftarrow{\T}[\infty..r_2]Y_2v$ are cube free. It remains to use Lemma~\ref{l:tmred}: take big enough $r_1', r_2'$ and find a word $w$ such that $\T[r_1..r_1']w\overleftarrow{\T}[r_2'..r_2]$ is a factor of $\T$ and a transition word for the pair $(uY_1, Y_2v)$; the uniform recurrence of $\T$ ensures that the word $w$ can be found in finite time. Thus $Y_1\T[r_1..r_1']w\overleftarrow{\T}[r_2'..r_2]Y_2$ is the transition word for the pair $(u,v)$, so Problem~5 is solved.

Once again, it is clear that some steps of the above solution can be significantly sped up, so it would be nice to finally get a polynomial-time algorithm for Problem~5 (and thus for Problems 1, 2, 4 as well). From the experimental study we learned that if a length-$n$ cube-free word is not right extendable, then likely not only all its right contexts have the length $O(\log n)$, but the number of such contexts is $O(\log n)$. The proof of this fact would lead to a linear-time solution of Problem~1. 


Another obvious continuation of the current research is the study of the same problems for other power-free languages. One line is to use Thue-Morse words to solve Problems~2,~4,~and~5 for other binary power-free languages. For example, we are able to extend the results of Section~\ref{ss:red} to $\alpha$-power-free binary words for any $\alpha\in(5/2,3]$, changing only some constants. Another line is to obtain similar results for ternary square-free words, in the absence of such a strong tool as Thue-Morse words.

\bibliographystyle{plainurl}
\bibliography{my_bib}

\end{document}